\documentclass[10pt]{article}

\usepackage[T1]{fontenc}
\usepackage[latin1]{inputenc}
\usepackage[english]{babel}
\usepackage{graphicx, amssymb, amsmath,amsthm, latexsym}
\usepackage{amscd}
\usepackage{latexsym}
\usepackage{tikz}
\usepackage[textwidth=15cm,textheight=20cm]{geometry}

\newtheorem{theorem}{Theorem}
\newtheorem{corollary}[theorem]{Corollary}
\newtheorem{definition}[theorem]{Definition}
\newtheorem{proposition}[theorem]{Proposition}
\newtheorem{lemma}[theorem]{Lemma}

\newtheorem{observation}[theorem]{Observation}

\newcommand{\dir}[1]{\overrightarrow{#1}}
\newcommand{\Tr}{\dir\triangleleft}
\newcommand{\SEP}{\dir{\gamma^{\text{\tiny{S}}}}}
\newcommand{\id}{\dir{\gamma^{\text{\tiny{ID}}}}}

\begin{document}
\title{Characterizing extremal digraphs for identifying codes and extremal cases of Bondy's theorem on induced subsets\thanks{This research is supported by the ANR Project IDEA {\scriptsize $\bullet$} {ANR-08-EMER-007},  2009-2012.}}

%\title{Extremal problems on identifying codes in digraphs and Bondy's theorem on induced subsets\thanks{This research is supported by the ANR Project IDEA {\scriptsize $\bullet$} {ANR-08-EMER-007},  2009-2012.}}

%\titlerunning{Characterizing extremal digraphs for identifying codes and Bondy's theorem}

\author{Florent~Foucaud  \and
        Reza~Naserasr \and
        Aline~Parreau
}

\maketitle

\begin{abstract}
An identifying code of a (di)graph $G$ is a dominating subset $C$ of the vertices of $G$
such that all distinct vertices of $G$ have distinct (in)neighbourhoods within $C$.
In this paper, we classify all finite digraphs which only admit their whole vertex set as
an identifying code. We also classify all such infinite oriented graphs. Furthermore, by
relating this concept to a well-known theorem of A.~Bondy on set systems, we classify the
extremal cases for this theorem.

\end{abstract}

\section{Introduction}

Identifying codes are dominating sets having the property that any two
vertices of the graph have distinct neighbourhoods within the
identifying code. As a consequence, they can be used to uniquely
identify or locate the vertices of a graph. Identifying codes have
been widely studied since the introduction of the concept
in~\cite{KCL98}. The theory is applied to problems such as  
fault-diagnosis in multiprocessor systems~\cite{KCL98} or
emergency sensor networks in facilities~\cite{DRSTU03}.
The concept of identifying codes is an extension of previous works on
locating-dominating sets, studied in e.g.~\cite{CHL07,RS84}.
Identifying codes have first been studied in undirected graphs,
but the concept has naturally been extended to
directed and oriented graphs~\cite{CGHLMM06,CHL02,S07}.

In this paper, we study classes of directed
graphs which only admit large minimum identifying codes, extending earlier
works on such problems for undirected graphs (see e.g.~\cite{CHL07,FGKNPV10,GM07}).
We also relate the problem of identifying codes in directed graphs to a set theoretic
problem first studied by A.~Bondy in~\cite{B72}.

Let us now give some notations and definitions. In the following and
unless otherwise stated, we will denote by $D$, a directed graph with
vertex set $V(D)$ and arc set $A(D)$. A directed graph will
conveniently be called a \emph{digraph}. An arc pointing from vertex
$u$ towards vertex $v$ will be denoted $\dir{uv}$. The arcs $\dir{uv}$
and $\dir{vu}$ are called \emph{symmetric}. A digraph without any
symmetric arcs will be called an \emph{oriented graph}.  For a vertex
$u$ of $D$, we denote by $N^-(u)$, the set of outgoing neighbours of
$u$, and by $N^+(u)$, the set of incoming neighbours of $u$. By
$d^+(u)$ and $d^-(u)$, respectively, we denote the in-degree and the
out-degree of the vertex $u$. Let $B_1^+(u)$ and $B_1^-(u)$ denote the
incoming and outgoing balls of radius~1 centered at $u$ (they include
$u$). A vertex of $D$ without any incoming arc is called a
\emph{source} in $D$. The \emph{underlying graph} of $D$ is the
graph obtained from $D$ by ignoring the directions of the arcs. We say
that $D$ is \emph{connected} if its underlying graph is connected.

In a digraph $D$, a \emph{dominating set} is a subset $C$ of $V(D)$ such that for all $u\in V(D)$,
$B_1^+(u)\cap C\neq\emptyset$.
A \emph{separating code} of $D$ is a subset $C$ of vertices such that each pair of vertices
$u, v$  of $D$ is \emph{separated}, that is: $B_1^+(u)\cap C\neq B_1^+(v)\cap
C$. If $C$ is both a dominating set and a separating code, it is called an \emph{identifying code} of $D$. 

Two vertices $u$ and $v$ are called \emph{twins} in a digraph $D$ if $B_1^+(u)=B_1^+(v)$. 
It is easy to observe that a digraph admits a separating code and, therefore, an identifying code,
if and only if it has no pair of twins ($C=V(D)$ is one such code in this case).
Such digraphs are called \emph{twin-free}. Note that if two vertices $u$ and $v$ are twins,
necessarily the symmetric arcs $\dir{uv}$ and $\dir{vu}$ must exist. As a consequence,
oriented graphs are always twin-free.
 
In this paper, we will also consider separating codes in undirected
bipartite graphs.  Given a bipartite graph $G=(S\cup T,E)$ with $S$
and $T$ being two parts, an \emph{$S$-separating code} of $G$ is a
subset $C$ of vertices in $T$ such that $N(u)\cap C\neq N(v)\cap C$
for all distinct vertices $u$ and $v$ of $S$. If moreover, every
vertex of $S$ is also dominated by some vertex of $C$, then $C$ is
called an \emph{$S$-discriminating code} of $G$. When $G$ admits an
$S$-discriminating code then we say that $G$ is
\emph{$S$-identifiable}. The concept of $S$-separating codes has
also been studied under the name of \emph{test collection problem},
see e.g.~\cite{MS85}. The concept of $S$-discriminating codes was
introduced in~\cite{CCCHL08}. We note that to find an identifying
code, respectively a separating code, of a graph or a digraph one
could form its incidence bipartite graph and search for an
$S$-discriminating or an $S$-separating code in this incidence
bipartite graph.

One of the most natural questions in the study of separating or identifying codes is to find 
the smallest such set for a given graph.
We denote by $\SEP(D)$ and by $\id(D)$, the minimum sizes of a separating code and of
an identifying code of $D$, respectively. It is observed that $\SEP(D) \leq \id(D) \leq \SEP(D)+1$.
Determining the value of $\id(D)$ was proven to be NP-hard
in~\cite{CHL02}, even for strongly connected bipartite oriented graphs and for bipartite
oriented graphs without directed cycles.

The theory of identifying codes is studied mostly for undirected graphs, which can be viewed as
symmetric digraphs. The following is one of the first general bounds:

\begin{theorem}[\cite{B01,GM07}]\label{thm:existence}
For any symmetric finite twin-free digraph $D$ having at least one arc, we have  $\id(D) \le |V(D)|-1$.
\end{theorem}

Undirected graphs achieving this bound have been characterized in~\cite{FGKNPV10}.
For the directed analogue of this classification problem, the class of digraphs satisfying
$\id(D)=|V(D)|$ is already rich.
In this paper we classify all such finite digraphs. This is done in Section~\ref{sec:finite}.
A simple corollary of our work is an extension of Theorem~\ref{thm:existence}
which says that every digraph with at least one symmetric arc satisfies: $\id(D) \le |V(D)|-1$.
In Section~\ref{sec:infinite} we classify all infinite oriented graphs in which the whole
vertex set is the only identifying code.  In Section~\ref{sec:Bondy}, we discuss the relationship
between this result and the following  theorem of A.~Bondy on "induced subsets":

\begin{theorem}[Bondy's theorem \cite{B72}]\label{Bondy}
Let  ${\mathcal A}=\{A_1,A_2, \cdots, A_n\}$ be a collection of $n$ distinct subsets of an
$n$-set $X$. Then there exists an element $x$ of $X$ such that the sets
$A_1-x,A_2-x, \cdots, A_n-x$ are all distinct.
\end{theorem}

Here $A_i-x$ could be the empty set. Though, to be distinct, at most one of them could be
empty. In Section~\ref{sec:Bondy} we classify all the set systems in which for any possible
choice of $x$ in Bondy's theorem we will have $A_i-x=\emptyset$ for some $i$.

Since its publication in 1972, this theorem has received a lot of
attention and various proofs of different nature have been provided
for it, we refer to \cite{BF92,B86,B72,L79,W00} for some interesting
proofs.

\section{Finite extremal digraphs}\label{sec:finite}

In this section, we classify all finite digraphs in which the whole
vertex set is the only identifying code. We begin with the following definitions of operations on
digraphs, illustrated in Figure~\ref{fig:operations}.

Given two digraphs $D_1$ and $D_2$ on disjoint sets of vertices, we define
the \emph{disjoint union} of $D_1$ and $D_2$, denoted $D_1 \oplus D_2 $, to
be the digraph whose vertex set is $V(D_1)\cup V(D_2)$ and whose arc set
is $A(D_1)\cup A(D_2)$.

Given a digraph $D$ and a vertex $x$ not in $V(D)$ we define $x
\dir\triangleleft(D)$ to be the digraph with vertex set $V(D)\cup
\{x\}$ whose arcs are the arcs of $D$ together with each arc
$\dir{xv}$ for every $v\in V(D)$.

\begin{figure}[ht]
\centering
\scalebox{0.6}{
\begin{tikzpicture}
\draw[line width=1pt, draw = black] (-9,0) ellipse(1cm and 2cm) node{\Large $D_1$};
\draw[line width=1pt, draw = black] (-6,0) ellipse(1cm and 2cm) node{\Large $D_2$};
\path (-7.5,-3) node (m) {\Large $D_1\oplus D_2$};

\path (-0.2,0) node[draw,shape=circle] (a) {\Large $x$};
\draw[line width=1pt, draw = black] (4,0) ellipse(1cm and 2cm) node{\Large $D$};
\draw[->,>=latex, thick] (a) -- +(3.7,-1.8);
\draw[->,>=latex, thick] (a) -- +(3.3,-1);
\draw[->,>=latex, thick] (a) -- +(3.2,0);
\draw[->,>=latex, thick] (a) -- +(3.3,1);
\draw[->,>=latex, thick] (a) -- +(3.7,1.8);
\path (1.5,-3) node (n) {\Large $x\Tr(D)$};
\end{tikzpicture}}
\caption{Illustrations of the two operations}
\label{fig:operations}
\end{figure}
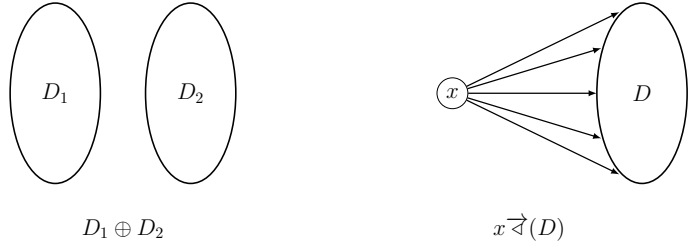

\begin{definition}\label{def:family}
Let $(K_1, \oplus, \dir\triangleleft )$ be the closure of the one-vertex graph $K_1$ with respect to the operations
$\oplus$ and $\dir\triangleleft$. That is, the class of all graphs that can be built
from $K_1$ by repeated applications of $\oplus$ and $\dir\triangleleft$.
\end{definition}

The \emph{transitive closure} of a digraph $D$ is a digraph obtained
from $D$ by adding the arc $\overrightarrow{xy}$ whenever there is a
directed path from $x$ to $y$. A \emph{rooted oriented tree} is an
oriented tree with a specific vertex $v$ (called the root) such that
for every other vertex $u$ the path connecting $u$ to $v$ is a
directed path from $v$ to $u$. In such a tree, given an arc
$\overrightarrow{xy}$, we say that $x$ is the father of $y$ and that
$y$ is a child of $x$. The vertices with a directed path to $x$ are
called \emph{ancestors} of $x$ and vertices with a directed path from
$x$ are \emph{descendents} of $x$. A \emph{rooted oriented forest} is
a disjoint union of rooted oriented trees.

By these definitions it is easy to check that:

\begin{observation}\label{TransitveClosuresOfTrees}
 Every element of $(K_1, \oplus, \dir\triangleleft )$ is the transitive 
closure of a rooted oriented forest.
\end{observation}

For a digraph $D$ in $(K_1, \oplus,
\dir\triangleleft )$ let us denote by $\dir{F}(D)$
the rooted oriented forest whose transitive closure is $D$.

\begin{proposition}\label{TreeRepresentation}
 For every digraph $D$ in  $(K_1, \oplus, \dir\triangleleft )$ we have
 $\id(D)=|V(D)|$. Furthermore, if a vertex $x$ is a source in $D$, then $V(D)-x$ is
a separating code of $D$. Otherwise, the vertex $x$ and its father in $\dir{F}(D)$
are the only ones not being separated from each other by the set $V(D)-x$.
\end{proposition}

\begin{proof}
Let $C$ be an identifying code of $D$. Except for its roots, each vertex of the forest $\dir{F}(D)$ must be in $C$ in order to be separated
from its father. But the sources need also to be in $C$ in order to be dominated.
\qed\end{proof}

The next proposition is the directed analogue of Proposition 3 in \cite{FGKNPV10},
and since the proof is quite the same we omit it. 

\begin{proposition}\label{PropositionForInduction}
Let $D$ be a finite twin-free digraph, and let $S$ be a subset of vertices of $D$ such that $D-S$ is twin-free.
Then $\id(D)\leq \id(D-S)+|S|$.
\end{proposition}

Let $D$ be a twin-free digraph with vertex set $\{ x_1, x_2, \cdots, x_n \}$ and let 
$\mathcal A=\{B_1^+(x_1), B_1^+(x_2), \cdots, B_1^+(x_n)\}$.
Then  $(\mathcal A, V(D))$ form a set system satisfying the conditions of Bondy's theorem.
Therefore we have:

\begin{proposition}\label{PropositionNMinusOne}
Let $D$ be a finite twin-free digraph. Then $\SEP(D)\leq |V(D)|-1$.
\end{proposition}

The following corollary of the previous proposition will also be needed.

\begin{proposition}\label{Eq:idcode-sepcode}
In a finite twin-free digraph $D$, $\id(D)=|V(D)|$ if and only if $\SEP(D)=|V(D)|-1$ and for every
minimum separating code of $D$, there is a vertex which is not dominated.
\end{proposition}

\begin{proof}
Suppose $\id(D)=|V(D)|$. By Proposition~\ref{PropositionNMinusOne} we know that $\SEP(D)\leq |V(D)|-1$.
On the other hand, $\id(D)\leq \SEP(D)+1$ thus $\SEP(D)=|V(D)|-1$. Now, if all the vertices where dominated by
some minimum separating code of $D$, that code would also be identifying, a contradiction.

Conversely, if $\SEP(D)=|V(D)|-1$ and for any minimum separating code of $D$, there is a vertex
which is not dominated, we are forced to take all the vertices in any identifying code in order
to get the domination property.
\qed\end{proof}

The following theorem shows that the family of Definition~\ref{def:family} is
exactly the class of finite digraphs in which the whole vertex set is
the only identifying code.

\begin{theorem}\label{M(D)=|V(D)|}
Let $D$ be a finite twin-free digraph. If  $\id(D)=|V(D)|$ then $D \in (K_1, \oplus, \dir\triangleleft )$.
\end{theorem}

\begin{proof}
Assume $D$ is the smallest digraph for which $\id(D)=|V(D)|$ but 
$D \notin (K_1, \oplus, \dir\triangleleft )$.
We consider two cases:

\textbf{Case a.} Assume that there exists a vertex $x$ of $D$ such that $x$ has no outneighbours.
Then $D-x$ is a twin-free graph. By Proposition~\ref{PropositionForInduction} we have 
$\id(D-x)=|V(D-x)|$ and, therefore, by the minimality of $D$, 
$D-x \in  (K_1, \oplus, \dir\triangleleft )$. Thus $\dir{F}(D-x)$ is well-defined. 
Since $V(D)-x$ is not an identifying code of $D$, in $V(D)-x$ either there is a vertex $y$ which is not separated
from $x$, or $x$ is not dominated.

If $x$ is not dominated, then $x$ is an isolated vertex and, therefore,
$D$ is the disjoint union of two members of $(K_1, \oplus, \dir\triangleleft )$, 
hence $D$ is also in the family. So there is a vertex $y$ which is not separated
from $x$. Therefore, $N^+(x)=N^+(y)\cup\{y\}$. Then $D$ is the transitive closure of
the oriented tree built from $\dir{F}(D-x)$ by adding $x$ as a child of $y$, a contradiction.

\textbf{Case b.} Every vertex has at least one outneighbour.
By Proposition~\ref{PropositionNMinusOne} we know that there exists a vertex $x$ such that $D-x$ is twin-free
and as in the previous case $\dir{F}(D-x)$ is well-defined.

Since every vertex, in particular, every leaf $t$ of $\dir{F}(D-x)$ has an outneighbour in $D$, we must have 
$\overrightarrow{tx}\in A(D)$. Thus $d^+(x)\geq 1$ and, therefore, $V(D)-x$ is a dominating set. But since
it is not an identifying code there is a vertex $y\neq x$ which is not separated from $x$ by $V(D)-x$, i.e.,
$N^+(x)=N^+(y)\cup\{y\}$.

We now claim that $y$ is the only leaf of $\dir{F} (D-x)$. That is
because if $t\neq y$ is a leaf then $t\in N^+(x)$ so $t\in N^+(y)$
which is a contradiction. But there has to be at least one leaf in
$\dir{F} (D-x)$ thus $y$ is the only leaf in $\dir{F} (D-x)$ and,
therefore, $\dir{F}(D-x)$ is a path.  Now since $d^-(x)>0$, there is a
vertex in $N^-(x)$. We have $y\notin N^-(x)$ since otherwise we would
have $N^+(x)=N^+(y)$. 

First, assume that there exists a vertex $t\in N^-(x)$ such that the
father of $t$ in $\dir{F} (D-x)$ is not in $N^-(x)$. We claim that
$C=V(D)-t$ is an identifying code. Indeed, $x$ is the only vertex
dominated by all vertices of $C$. Vertex $t$ and its father are
separated by $x$. Finally, each other pair of vertices from $V(D)-x$
is separated by the one which is a descendant of the other in $\dir{F}
(D-x)$.

Now, assume that there is no vertex in $N^-(x)$ with its father in
$\dir{F} (D-x)$ not belonging to $N^-(x)$. Let $r$ be the root of
$\dir{F} (D-x)$. In particular, $r\in N^-(x)$. We claim that
$C=V(D)-r$ is an identifying code. Indeed, $x$ is the only vertex
dominated by all vertices of $C$. Each pair of vertices from $V(D)-x$
is separated by the one which is a descendant of the other in $\dir{F}
(D-x)$. Finally, $r$ is the only vertex which is dominated only by
$x$.
\qed\end{proof}

Noticing that digraphs in  $(K_1, \oplus, \dir\triangleleft )$ have no symmetric arcs, we have the
following extension of Theorem~\ref{thm:existence}:

\begin{corollary}
Let $D$ be a finite twin-free digraph having some symmetric arcs. Then $\id(D)\leq |V(D)|-1$.
\end{corollary}

\section{Extremal infinite oriented graphs}\label{sec:infinite}

In this section, we consider the case of infinite digraphs in which the whole vertex set is the only identifying code.
To avoid set theoretic problems we only consider infinite graphs on a countable set of vertices.
Thus in this section the considered directed or oriented graphs have a finite or countable set of vertices.

By the characterization of the family of simple infinite graphs needing their whole vertex set
to be identifed in~\cite{FGKNPV10}, we already have a rich family of extremal symmetric digraphs with
respect to identifying codes. The family of all such directed graphs seems to be too rich to characterize.
In this section we provide such characterization for the class of all oriented graphs. We begin with the
following definitions:

A connected (and possibly infinite) oriented graph $D$ is a
\emph{finite-source transitive tree} ($fst$-tree for short) if:
\begin{enumerate}
\item[(1)] for each vertex $x$ of $D$, $B_1^+(x)$ induces 
the transitive closure of a finite directed path, $P_x$, with $x$ as its end vertex, and
\item[(2)] for each pair $x,y$ of vertices, there is a vertex $z\in V(D)$ such that $P_x\cap P_y=P_z$.
\end{enumerate}
Note that in an $fst$-tree $D$, since each path $P_x$ is finite,
point~(2) of the definiton implies that for any pair $x,y$ of
vertices, $P_x$ and $P_y$ begin with the same vertex. Hence all paths
$P_x$ begin with the same vertex, which is the unique source of $D$.

An oriented graph $D$ is an {\em infinite-source transitive tree} ($ist$-tree for short) 
if:
\begin{enumerate}
\item[(1)] for all vertices $x$ of $D$, $B_1^+(x)$ induces the
  transitive closure of an infinite path (we denote it by $P_x$), and
\item[(2)] for any pair $x,y$ of vertices of $D$, there is a vertex $z\in
  V(D)$ such that $P_x\cap P_y = P_z$.
\end{enumerate}
Note that an $ist$-tree has no source vertex but one can imagine
infinity as its source.

Finally we say that an oriented graph $D$ is a \emph{source transitive
  tree} if it is either an $fst$-tree or an $ist$-tree, see
Figure~\ref{FST-ISTtrees} for examples. For each pair $x,y$ of
vertices of a source transitive tree $D$, paths $P_x$ and $P_y$ share
the vertices of a path $P_z$ which includes the ``beginning'' of both
$P_x$ and $P_y$. Hence, all arcs of $D$ can be oriented in the same
direction. Moreover, there cannot be any cycle in the union of all
paths $P_x$. This implies that each source transitive tree $D$ is the
transitive closure of a finite or infinite ``rooted'' oriented tree
(even if an $ist$-tree has no properly defined root vertex, one can
regard infinity as its root) which we call the \emph{underlying tree}
of $D$.  Notice that the collection of $fst$-trees on a finite set of
vertices is exactly the set of connected elements of $(K_1, \oplus,
\dir\triangleleft )$.

\begin{figure}[ht]
\begin{center}
\resizebox{8cm}{!}{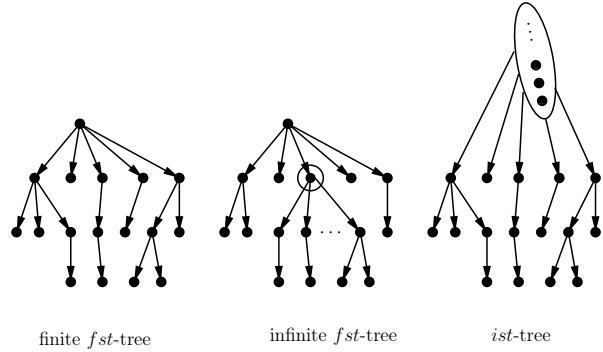}
\end{center}
\caption{Underlying trees of source transitive trees}\label{FST-ISTtrees}
\end{figure}

\begin{proposition}\label{prop:tree}
The only identifying code of a source transitive tree is its whole set of vertices.
\end{proposition}

\begin{proof}
Let $D=(V,A)$ be a source transitive tree and $x$ be any vertex of $D$. 
If $x$ is a source (it can only happen if $D$ is an $fst$-tree),
then $x$ must be in any identifying code of $D$ in order to be dominated. If $x$ is not a source,
then to separate $x$ from its father in the underlying tree of $D$, $x$ itself must be in any identifying
code.
\qed\end{proof}

We are now ready to build the whole family of oriented graphs that need their whole vertex
set to be identified. To this end given any oriented graph $H$ we first build the family $\Psi(H)$
of extremal oriented graphs as follows:

For each vertex $x$ of $H$ if $x$ is a source of $H$, then we assign an $fst$-tree $T_x$ to $x$.
If $x$ is not a source, then we assign an $ist$-tree $T_x$ to $x$. The choice of $T_x$ is free but each
$T_x$ has its distinct set of vertices. For each arc $\overrightarrow{xy}$ of $H$ we also associate
a subset $V_{ \overrightarrow{xy}}$ of $V(T_x)$ (the choice of $V_{ \overrightarrow{xy}}$ is also free).
We now build a member of $\Psi(H)$ by taking $\cup V(T_x)$ as the vertex set, arcs of $T_x$ are also
arcs of the new graph and, furthermore, for any $z \in V_{ \overrightarrow{xy}}$ and any $t\in V(T_y)$, we
add an arc $\overrightarrow{zt}$.

\begin{proposition}\label{prop:constructioninfini}
Given an oriented graph $H$, any digraph $D$ in $\Psi(H)$ can only be identified by its
whole vertex set.
\end{proposition}

\begin{proof}
The sources of $D$ are exactly the sources of the $fst$-trees $T_x$ for source-vertices $x$ of $H$
and need to be in any identifying code in order to be dominated.
If a vertex $u$ of $D$ is not a source, then it is in an $fst$ or $ist$-tree $T_x$ and there is,
like in the proof of Proposition~\ref{prop:tree}, a vertex $v$ of $T_x$ such that
$B_1^+(u)\cap V(T_x)=(B_1^+(v)\cap V(T_x))\cup\{u\}$ ($v$ simply is the father of $u$ in the underlying tree of $T_x$).
By our construction, any incoming neighbour of $u$ not in $T_x$ is also an incoming neighbour of
$v$ so  $B_1^+(u)=B_1^+(v)\cup\{u\}$ and $u$ must be in any identifying code of $D$.
\qed\end{proof}

\begin{theorem}\label{thm:orientedinfinite} 
Let $D$ be an infinite twin-free oriented graph. Then
a proper subset of $V(D)$ identifies all pairs of vertices of $D$ unless
$G\in \Psi(H)$ for some finite or infinite oriented graph $H$.
\end{theorem}

\begin{proof}
Let $D$ be an infinite oriented graph that needs its whole vertex set to be identified.
Let $x$ be a vertex of $D$. The set $V(D)-\{x\}$ is not an identifying code.
Either $x$ is not dominated and so, $x$ is a source or there is a pair of vertices, say $u$ and $v$,
such that $B_1^+(v)=B_1^+(u)\cup\{x\}$. If $x\neq v$, we must have $\dir{uv} \in A$ and $\dir{vu} \in A$.
Since $D$ has no symmetric arc, this is not possible so, necessarily, $x=v$, $B_1^+(x)=B_1^+(u)\cup\{x\}$
and $u$ is the only vertex such that $B_1^+(x)=B_1^+(u)\cup\{x\}$.
So for any vertex $x$ of $D$ which is not a source there is a unique vertex we call $x_{-1}$ such that
$B_1^+(x)=B_1^+(x_{-1})\cup\{x\}$. We may repeat this argument on $x_{-i}$ to get $x_{-i-1}$ for $i=1,2, \cdots $
as long as $x_{-i}$ is not a source. This will result in a well defined set $\{ \cdots, x_{-i}, \cdots, x_0=x\}$
which induces a transitive closure of a finite or infinite path, which we denote by $P_x$
(if $x$ is a source itself, then $P_x=\{x\}$).

Assume that for two vertices $x$ and $y$, $P_x\cap P_y \neq \emptyset$. Then let $x_i$ be the
first (in the order defined by the path $P_x$) vertex of $P_x$ in $P_x\cap P_y$. We have $P_x\cap P_y=P_{x_i}$.

We now define an equivalence relation on the vertices of $D$: $x\equiv y$ if and only if $P_x\cap P_y \neq \emptyset$.
This gives us the equivalence class of $x$: $T_x=\{y \in V(D)\vert x\equiv y\}$. 
The set of vertices of $T_x$ induces either an $fst$-tree or an $ist$-tree in $D$.
Furthermore, if $u\notin T_x$, then $\overrightarrow {uv}$ is an arc of $D$ for either every $v\in V(T_x)$
or no $v \in V(T_x)$.
In fact, if there is an arc $\dir{uy}$ in $D$ for $y\in T_x$, then if $z \in T_x$ we have $P_y\cap P_z=P_{y_{i}}$. 
But $B_1^+(y)=B_1^+(t)\cup \{y,y_{-1},...,y_{i+1}\}$
so $u\in B_1^+(t)$. Now, $B_1^+(t)\subset B_1^+(z)$ so $u\in B_1^+(z)$ and $\dir{uz}$ is also an arc of $D$.

We construct a graph $H$ as follows: the vertices of $H$ are the equivalence classes $T_x$ and there is an arc 
$\overrightarrow{T_xT_y}$ if there is an arc $\overrightarrow{uv}$ of $D$ such that $u\in V(T_x)$ and $v\in V(T_y)$.
It is now clear that $D\in \Psi(H)$. 
\qed\end{proof}

We conclude this section by the following remarks:

\begin{itemize}
\item This proof also works for the classification of finite oriented graphs with $\id(D)=|V(D)|$.
But the characterization of Theorem~\ref{M(D)=|V(D)|} is for all digraphs, thus it is a stronger statement.

\item The oriented graphs for which the only separating set is their whole vertex set are 
graphs in $\Psi(H)$ as long as $H$ has no source vertex.
\end{itemize}

\section{An application to Bondy's Theorem}\label{sec:Bondy}

In this section, unless specifically mentioned, a set system is a pair
$(\mathcal A, X)$ with $X$ being any set of size $n$
and $\mathcal A$ being a collection of $n$ distinct subsets of $X$.
When applying Bondy's theorem to a set system $(\mathcal{A},X)$ where all subsets in 
$\mathcal{A}$ are distinct \emph{and nonempty}, it is a natural request to be able to
choose an element $x$ of $X$ such that no subset $A_i-x$ is the empty set.
This is not always possible, just consider the set system ${\mathcal A}$ consisting of all
singletons of $X$. Such set systems will be called \emph{extremal}. More precisely,
an extremal set system is a set system $ (\mathcal{A},X)$ in which elements of
$\mathcal{A}$ are all distinct and nonempty, and where for any element $x$
of $X$ either there is an element $A_i\in \mathcal{A}$ with $A_i-x=\emptyset$
or there is a pair $A_i, A_j\in \mathcal{A}$ such that $A_i-x=A_j-x$.
In this section we characterize all such extremal cases.

We would like to mention that almost any proof of Bondy's theorem (e.g., see \cite{B72,B86,BF92,L79})
works for an extension
of this theorem in which we are allowed to have more elements in $X$ than $\mathcal A$.
We then look for a subset $X'\subset X$ of size $|X|-|\mathcal A|+1$  such that all the induced
sets $A_i-X'$ are distinct. The following proposition is now an easy consequence of this
general version of Bondy's theorem.

\begin{proposition}\label{emptyset}
Let $(\mathcal A,X)$ be a set system with $\vert X \vert > \vert \mathcal A \vert$ 
where all the subsets in $\mathcal A$ are distinct and nonempty. 
There is a subset $X'$ of $X$  of size $\vert X \vert - \vert \mathcal A \vert$ such that
all the subsets $A\cap(X- X')$ for $A\in\mathcal A$ are nonempty and distinct.
\end{proposition}

\begin{proof}
Let $X_0$ be the subset of size $\vert X \vert - \vert \mathcal A \vert+1$
found by extended version of Bondy's theorem i.e., all the subsets $A\cap(X- X_0)$ are distinct.
Thus there is at most one subset $A_0$ such that $A_0\cap(X- X_0)$ is empty. 
Let $x_0\in X_0$ be an element of $A_0$. Then $X'=X_0- \{x_0\}$ satisfies the proposition. 
\qed\end{proof}

To achieve our goal of characterizing the extremal set systems, we will need a few more definitions.
Given a set system $(\mathcal A,X)$ we define its \emph {incidence bipartite graph} 
$B(\mathcal A,X)$ to be the bipartite graph with $\mathcal A\cup X$ as its vertex set where $\mathcal A$ and
$X$ form the two parts and in which vertex $A_i\in\mathcal A$ is adjacent to vertex $x_j\in X$ if and only if
$x_j\in A_i$. Bondy's theorem can now be restated as follows:

\begin{theorem}
 Let $G=(S\cup T, E)$ be an $S$-identifiable bipartite graph with $|S|=|T|$.
Then there exists an $S$-separating code $C\subseteq T$ of size at most $|S|-1$.
\end{theorem}

As in the case of set systems one can also define the \emph{incidence bipartite graph} of a digraph.
To this end given a digraph $D$ on a vertex set $\{x_1,x_2, \cdots, x_n\}$ we define $B(D)$
to be the bipartite graph on $S=\{x_1,x_2, \cdots ,x_n\}$ and $T=\{x'_1,x'_2, \cdots, x'_n\}$
with $x_i$ being adjacent to $x'_j$ if either $\overrightarrow{x_ix_j} \in A(D)$ or $i=j$.
The latter condition of adjacency implies that for a bipartite graph to be the incidence bipartite graph of some
digraph it must admit a perfect matching. Below we prove that this is also a sufficient condition.

\begin{lemma}\label{BipartiteToDigraph}
A bipartite graph $G=(S\cup T,E)$ is the incidence bipartite graph of some digraph $D$
if and only if $|S|=|T|$ and $G$ admits a perfect matching $\varphi:S\rightarrow T$.
\end{lemma}

\begin{proof}
It follows from the definition of the incidence bipartite graph of a digraph that $G$ must
have both parts of the same size and admits a perfect matching which matches the two copies of each vertex. 

Now, given a bipartite graph $G$ with parts $S$ and $T$
of equal size together with a perfect matching $\varphi$, one can construct a digraph $D$ with vertex
set $S$ in the following way. For each pair $x,y$ of vertices if $x\varphi(y) \in E(G)$, then $\overrightarrow{xy}$
is an arc of $D$. The constructed digraph has $G$ as its incidence bipartite graph.
\qed\end{proof}

An example of the correspondence between a bipartite graph with a
perfect matching (thick edges) and a digraph is given in Figure~\ref{fig:example}.

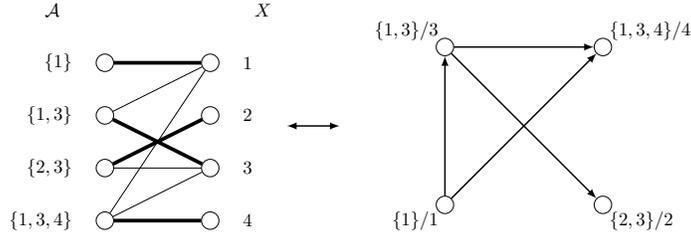
\begin{figure}
\begin{center}
\scalebox{0.7}{\begin{tikzpicture}

\path (-1,4) node {$\mathcal{A}$};
\path (0,3) node[draw,shape=circle] (s1) {};
\draw (s1) node[left=0.5cm] {$\{1\}$};
\path (0,2) node[draw,shape=circle] (s13) {};
\draw (s13) node[left=0.5cm] {$\{1,3\}$};
\path (0,1) node[draw,shape=circle] (s23) {};
\draw (s23) node[left=0.5cm] {$\{2,3\}$};
\path (0,0) node[draw,shape=circle] (s134) {};
\draw (s134) node[left=0.5cm] {$\{1,3,4\}$};

\path (3,4) node {$X$};
\path (2,3) node[draw,shape=circle] (e1) {};
\draw (e1) node[right=0.5cm] {$1$};
\path (2,2) node[draw,shape=circle] (e2) {};
\draw (e2) node[right=0.5cm] {$2$};
\path (2,1) node[draw,shape=circle] (e3) {};
\draw (e3) node[right=0.5cm] {$3$};
\path (2,0) node[draw,shape=circle] (e4) {};
\draw (e4) node[right=0.5cm] {$4$};

\draw (s1) -- (e1) -- (s13) -- (e3) -- (s134) -- (e4)
      (e1) -- (s134)
      (e2) -- (s23) -- (e3);
\draw[line width=2pt] (e1) -- (s1)
                        (e2) -- (s23)
                        (e3) -- (s13)
                        (e4) -- (s134);
\end{tikzpicture}}
\scalebox{0.7}{\begin{tikzpicture}

\path (0,0) node[draw,shape=circle] (e1) {};
\draw (e1) node[below left] {$\{1\}/1$};
\path (3,0) node[draw,shape=circle] (e2) {};
\draw (e2) node[below right] {$\{2,3\}/2$};
\path (0,3) node[draw,shape=circle] (e3) {};
\draw (e3) node[above left] {$\{1,3\}/3$};
\path (3,3) node[draw,shape=circle] (e4) {};
\draw (e4) node[above right] {$\{1,3,4\}/4$};

\draw[->, >=latex, thick] (e3) -- (e2);
\draw[->, >=latex, thick] (e3) -- (e4);
\draw[->, >=latex, thick] (e1) -- (e3);
\draw[->, >=latex, thick] (e1) -- (e4);
\draw[<->, >=latex, thick] (-3,1.5) -- +(1,0);

\end{tikzpicture}}
\caption{Digraph of a bipartite graph}\label{fig:example}
\end{center}
\end{figure}

We are now ready to achieve our goal of classifying the extremal set systems in Bondy's theorem
when $|\mathcal{A}|=|X|$.

\begin{theorem}\label{ApplicationToBondy}
A set system  $(\mathcal{A},X)$  with $|\mathcal{A}|=|X|$ is extremal
if and only if its incidence bipartite graph $B(\mathcal{A},X)$ is 
the incidence bipartite graph of a digraph in $(K_1, \oplus, \dir\triangleleft)$.
\end{theorem}

\begin{proof}
If $B(\mathcal{A},X)$ is the incidence bipartite graph of a member $D$
of $(K_1, \oplus, \dir\triangleleft )$, then by
Theorem~\ref{M(D)=|V(D)|} we have $\id(D)=|V(D)|$ thus by
Proposition~\ref{Eq:idcode-sepcode} with any separating code of size
$|V(D)|-1$ there must be a vertex which is not dominated. Any
separating code $V(D)-x$ of $D$ corresponds to the choice of $x$ in
Bondy's theorem, and leaves some vertex in $D$ undominated --- that
is, there exists $A_i\in\mathcal{A}$ such that
$A_i-x=\emptyset$. Hence $(\mathcal{A},X)$ is extremal.

For the other direction, we distinguish two cases.

\textbf{Case a.} $B(\mathcal A, X)$ admits a perfect matching. Then
the directed graph $D$ built from $B(\mathcal A, X)$ using this
matching as in the method of Proposition~\ref{Eq:idcode-sepcode} is
such that $\id(D)=|V(D)|$. Thus by Theorem~\ref{M(D)=|V(D)|}, $D\in
(K_1, \oplus, \dir\triangleleft)$ and we are done.

\textbf{Case b.} $B=B(\mathcal A, X)$ does not contain a perfect matching. 
Then by Hall's marriage theorem (see e.g.~\cite{B86}), there is a subset $X'$
of $X$ such that $|N_B(X')|<|X'|$. For $A_i\in \mathcal A$ we define $A_i'=A_i\cap X'$
and $\mathcal {A'}=\{A_1'\cdots A'_{|\mathcal A|}\} -\emptyset$. 
Consider the set system $(\mathcal A',X')$ ($|\mathcal A'|< |X'|$). 
By Proposition~\ref{emptyset} there is an element
$x_0$ in $X'$ such that $A_i'-x_0$ are all nonempty and distinct as long as they induce distinct
elements in $\mathcal A'$. Now it is easy to check that $X-x_0$ induces  nonempty and distinct
elements on $\mathcal A$.
\qed\end{proof}

Using the previous theorem we can describe the extremal set systems  $(\mathcal{A},X)$
with $|\mathcal{A}|=|X|$ purely in the terminology of sets:

\begin{corollary}
A set system $(\mathcal{A},X)$ with $|\mathcal{A}|=|X|$ is extremal if and only if:
\begin{itemize}
\item $\bigcup A_i = X$
\item for any subset $A_i$ with at least two elements, there is an element $x$ of $A_i$ such that $A_i-x \in \mathcal A$
\end{itemize}

\end{corollary}

\begin{proof} 
If $D$ is a digraph of $(K_1, \oplus, \dir\triangleleft)$, it is easy to see that the set system 
corresponding to its incidence bipartite graph $B(D)$ has the properties.

Now, let $(\mathcal{A},X)$ be a set system having the properties of the corollary. 
Assume there is an element $x$ such that all $A_i-x$ are distinct and nonempty.
Take $A_i$ to be the smallest subset containing $x$, it exists because $\bigcup A_i = X$.
Then $A_i\neq \{x\}$ and so there is an element 
$y$ such that $A_i-y=A_j$. Then necessarily $x\neq y$ and $A_j$ is a smaller set than $A_i$ containing $x$. 
This is a contradiction. 
\qed\end{proof}

We conclude this section by relating this work to a similar problem
studied by Charon et al. in~\cite{CCCHL08}. We have previously seen that
Bondy's theorem can be stated in the language of separating codes in
bipartite graphs. Furthermore, the extremal case we have studied, where
a minimum separating code does always give an undominated vertex, is
equivalent to the one where the only minimum discriminating code consists in the
whole vertex set (in fact Proposition~\ref{Eq:idcode-sepcode} holds also
for separating and discriminating codes in bipartite graphs). 
Therefore Theorem~\ref{ApplicationToBondy} can be stated in the language of discriminating codes:

\begin{corollary}
  Given a bipartite graph $G=(S\cup T,E)$ with $|S|=|T|$ which is $S$-identifable,
  a minimum discriminating code 
  $C\subseteq T$ of $G$ has
  size $|S|$ if and only if $G$ is the
  incidence bipartite graph $B(D)$ of some digraph $D$ in $(K_1, \oplus, \dir\triangleleft)$.
\end{corollary}

\end{document}